\documentclass[10pt, a4paper, notitlepage]{article}

\usepackage{comments-hidden}
\usepackage{fullpage}
\usepackage{titling}

\newsavebox{\abstractbox}
\renewenvironment{abstract}
 {%
  \global\setbox\abstractbox=\vtop\bgroup
  \begin{center}\bfseries\abstractname\end{center}%
 }
 {\par\egroup}

\usepackage{authblk}
\usepackage{amsthm}
\usepackage{amssymb}
\usepackage{stmaryrd}
\usepackage{amsmath}
\usepackage{booktabs} 
\usepackage{thmtools}
\usepackage{xcolor}
\usepackage{xspace}
\usepackage{hyperref}
\usepackage{graphicx}
\usepackage{float}
\usepackage{comments}
\usepackage{enumerate}

\allowdisplaybreaks

\newtheorem{theorem}{Theorem}
\numberwithin{theorem}{section}
\newtheorem{lemma}[theorem]{Lemma}
\newtheorem{definition}[theorem]{Definition}
\newtheorem{corollary}[theorem]{Corollary}

\newtheorem{proposition}[theorem]{Proposition}

\title{Solving linear programs on factorized databases}

\author[1]{Florent Capelli}
\affil[1]{{Universit\'{e} de Lille} 
  }

\author[2]{Nicolas Crosetti}
\affil[2]{{Inria, Lille} 
}

\author[2]{Joachim Niehren}

\author[2]{Jan Ramon}

\newcommand{\void}[1] {}

\newcommand{\NP}{\mathsf{NP}}
\newcommand{\coNP}{\mathsf{coNP}}

\newcommand\ignore[1]{}

\newcommand{\calC}{\mathcal{C}}
\newcommand{\Ext}{\mathit{Ext}}

\newcommand{\andsym}{\times}
\newcommand{\orsym}{\uplus}

\newcommand{\andgate}{$\andsym$-gate}
\newcommand{\orgate}{$\orsym$-gate}
\newcommand{\cupgate}{$\cup$-gate}

\newcommand{\W} {W}

\newcommand{\om}[1][] {\omega_{#1}}

\newcommand{\val}[2] {#1(#2)}

\newcommand{\eweightings}[1][\circuit] {\assigntype{\edges{#1}}{\R_+}}
\newcommand{\tweightings}[1][] {\assigntype{\solcircuit{#1}}{\R_+}}

\newcommand{\solsym}{\mathit{rel}}

\newcommand{\solcircuit}[2][\circuit] {\solsym(#1_{#2})}
\newcommand{\soledge}[2][\circuit] {\solsym(#1, #2)}

\newcommand{\solproj}[3][\circuit]{\solsym(#1, #2, #3)}

\newcommand{\proj}[2][\tau] {#1_{\mid \var{#2}}}
\newcommand{\projt}[2] {#1_{\mid {#2}}}

\newcommand{\assigntype}[2] {#2^{#1}}

\newcommand{\R}{\mathbb{R}}

\newcommand{\rootedge} {o_r}

\newcommand{\proofnote}[1]{\text{\color{gray} \quad (#1)}}

\newcommand{\edge}[2] {\langle #1, #2\rangle} 
\newcommand{\prooftree}[2] {\mathsf{T}_{#1}(#2)}
\newcommand{\somept} {\mathsf{T}}

\newcommand{\ddCircuit} {$\{\uplus,\times\}$-Circuit\xspace} 
\newcommand{\ddcircuit} {$\{\uplus,\times\}$-circuit\xspace} 
\newcommand{\udcircuit} {$\{\cup,\times\}$-circuit\xspace} 

\newcommand{\circuit} {C}

\newcommand{\edges}[1] {\mathit{Edges(#1)}}
\newcommand{\nodes}[1] {\mathit{Nodes(#1)}}

\newcommand{\xdset}[3]{#1_{#2/#3}}

\newcommand{\tupleset}[2]{\xdset{S}{#1}{#2}}

\newcommand{\var}[1] {\mathsf{Attr}(#1)}
\newcommand{\dom} {\mathsf{D}}
\newcommand{\db} {\mathbb{D}}
\newcommand{\varx} {\mathsf{X}}

\newcommand{\disjointsolunion}[3] {\biguplus\limits_{#1 \in #2} \soledge[#3]{#1}}

\newcommand{\ingoing}[1] {\mathsf{In}(#1)}
\newcommand{\outgoing}[1] {\mathsf{Out}(#1)}

\newcommand{\sumlp}{CAS-LP}
\newcommand{\sxdname}{CAS variable}

\newcommand{\attrset}{\mathsf{X}}

\newcommand{\sxd}[2]{S_{#1, #2}}

\newcommand{\xdtuplesum}[3]{\sum\limits_{\substack{\tau \in #3 \\ \tau(#1)=#2}} \val{\om}{\tau}}

\newcommand{\sol}[1]{sol(#1)}

\begin{document}

\begin{abstract}

A typical workflow for solving a linear programming problem is to first write a
linear program parametrized by the data in a language such as Math GNU Prog or
AMPL then call the solver on this program while providing the data. When the data
is extracted using a query on a database, this approach ignores the underlying
structure of the answer set which may result in a blow-up of the size of the
linear program if the answer set is big. In this paper, we study the problem of
solving linear programming problems whose variables are the answers to a
conjunctive query. We show that one can exploit the structure of the query to
rewrite the linear program so that its size depends only on the size of the
database and not on the size of the answer set. More precisely, we give a
generic way of rewriting a linear program whose variables are the tuples in
$Q(D)$ for a conjunctive query $Q$ and a database $D$ into a linear program
having a number of variables that only depends on the size of a factorized
representation of $Q(D)$, which can be much smaller when the fractional
hypertree width of $Q$ is bounded.

\end{abstract}

\maketitle

\section{Introduction}
\label{sec:intoduction}

Computing the entire answer set of a database query is often only a first step
toward a more complex processing of the data usually involving different
aggregation tasks such as counting or computing average values. There has been a
successful line of research to design algorithms that could aggregate the
answers without explicitly solving the query first. For example, building on the
celebrated result of Yannakakis~\cite{Yannakakis} and its
extensions~\cite{gottlob_hypertree_2002}, Pichler and
Skritek~\cite{pichler_tractable_2013} have shown that one can count the number
of solutions of a quantifier-free acyclic (or of bounded hypertree width)
conjunctive query in polynomial time in the size of the database, which can be
significantly smaller than the size of the answer sets. In this paper, we are
interested in an other kind of aggregation problems: solving optimization
problems on the answers of a query. More precisely, we are interested in solving
linear programs whose variables are the answers of a database query. This kind
of problems is natural and it corresponds to the way most languages for
optimization problems such as GNU MathProg and AMPL~\cite{ampl} are designed.
They allow to write an optimization problem by separating the definition of its
constraints on parametrized sets -- that can alternatively be seen as database
tables -- and a part describing the actual content of these sets --
corresponding to \emph{the data} stored in the table (see the example files of
The AMPL Book~\cite{amplbook} for program examples). While these tools can
easily be interfaced with a DBMS to access data, they do not receive any
information concerning how the data was extracted from the database. In
particular, when the data is the result of a query, it may be structured in a
way that could be exploited by the solver to generate a smaller optimization
problem. Before going further into details, we illustrate this phenomenon with a
simple example showing how one can rewrite a linear program whose variables are
the solutions to a conjunctive query into a smaller equivalent one.

\subsection{Motivating example}
\label{sec:example}


\newcommand{\projNeeds}{projects}

\newcommand{\sciContributors}{researchers}
\newcommand{\expContributors}{developers}


\newcommand{\project}{pname}

\newcommand{\sciContributor}{rname}
\newcommand{\expContributor}{dname}

\newcommand{\sciExpertise}{field}
\newcommand{\expExpertise}{language}


\newcommand{\pjI}{p1}
\newcommand{\pjII}{p2}

\newcommand{\scI}{Alice}
\newcommand{\scII}{Bob}
\newcommand{\scIII}{Carol}
\newcommand{\scIV}{David}

\newcommand{\seI}{ML}
\newcommand{\seII}{DBs}

\newcommand{\ecI}{Eve}
\newcommand{\ecII}{Frida}
\newcommand{\ecIII}{Guy}

\newcommand{\eeI}{Python}


\newcommand{\projectTable}{
\begin{tabular}{l|l|l|l}
\projNeeds & \project & \sciExpertise & \expExpertise \\ \hline
    & \pjI  & \seI  & \eeI \\
    & \pjII & \seII & \eeI \\
\end{tabular}
}


\newcommand{\qname}{Q}
\newcommand{\domain}[1]{D_{#1}}

\newcommand{\exW}[2]{\omega^{#1}_{#2}}


\newcommand{\sciTable}{
\begin{tabular}{l|l|l}
\sciContributors & \sciContributor & \sciExpertise \\ \hline
    & \scI      & \seI  \\
    & \scII     & \seI  \\
    & \scIII    & \seI  \\
    & \scIV     & \seII \\
\end{tabular}
}


\newcommand{\expTable}{
\begin{tabular}{l|l|l}

\expContributors & \expContributor & \expExpertise \\ \hline
    & \ecI      & \eeI \\
    & \ecII     & \eeI \\
    & \ecIII    & \eeI \\
\end{tabular}
}


To illustrate how one can exploit the structure of the query to improve the
performances of the linear solver, we consider the following example. A
university is funding several applied research projects where it is needed for
researchers and developpers team up. Each project requires knowledge of a
specific field and expertise in a programming language from the assigned
researchers and developpers respectively. For simplicity, we assume here that
each project needs exactly one field and one programming language. Moreover the
researchers and developpers should work as pairs so for a given project the
assigned researchers should work the same amount of time as the developpers. The
goal of the university is to optimize the total time invested into the projects
while ensuring that no contributor is overworked. The data about the researchers
and projects is represented as follows:

\begin{itemize}

    \item $\projNeeds(\underline{\project}, \sciExpertise, \expExpertise)$ lists the projects, the field they are
        related to and the required programming language,
    \item $\sciContributors(\underline{\sciContributor}, \sciExpertise)$ lists the researchers and the fields
        they are knowledgeable in,
    \item $\expContributors(\underline{\expContributor}, \expExpertise)$ lists the developpers and the
        programming languages they are proficient in.
\end{itemize}

All the possible assignments of contributors to projects can be expressed with a simple conjunctive
query: 
\[ 
    \qname(p, r, d, f, l) = \projNeeds(p, f, l) \wedge \sciContributors(r, f) \wedge
    \expContributors(d, l).
\]
  
Given a project $p$, language $l$, field $f$, researcher $r$ and developper $d$,
we define $\om(p, r, d, f, l)$ to be the percentage of time allotted for $d$ and
$r$ to work on $p$, e.g. $\om(\pjI, \scI, \ecI, Databases, Python) = 42\%$ means
that $\scI$ and $\ecI$ will work on $\pjI$ as a pair for $42\%$ of their time
because of their respective skills in Databases and Python. The problem of
maximizing the overall time spent on the projects without overloaded the
contributors can be modelled as shown in Fig.~\ref{fig:lpex}.

\begin{figure}

\begin{align*}
    \text{maximize } & \sum\limits_{(p, r, d, f,l) \in \qname} \om(p, r, d, f, l) \\
    \text{subject to }  
        & \forall r \in \domain{\sciContributor}, 
            \sum\limits_{p,d,f,l} \om(p, r, d, f, l) \leq 100. \\
        & \forall d \in \domain{\expContributor}, 
            \sum\limits_{p,r,f,l} \om(p, r, d, f, l) \leq 100. \\
\end{align*}

\caption{Linear program example}
\label{fig:lpex}
\end{figure}

Observe that this linear program model is independent of the data stored in the
database which corresponds to the typical way of modelling linear programs using
languages like AMPL or GNU MathProg. However, given a database $\db$, this model
can be unfolded into a linear program whose size is proportional to the size of
$\qname(\db)$. In this case, as $\qname$ is the natural join of three tables,
$\qname(\db)$ may be of size $O(N^3)$ if we consider input relations of size
bounded by $N$. Thus, it will results in a linear program having $O(N^3)$
variables and $|\domain{\sciContributor}|+|\domain{\expContributor}|$
constraints.

It turns out that the structure of the query, which is acyclic, allows us to
rewrite the linear program into an equivalent one which has the same optimal
value but has only $O(N)$ variables and a comparable number of constraints. The
goal of this paper is to provide a generic way of doing this transformation when
the query is sufficiently well-structured. 

Observe that one could construct a similar example with any $k$ different types
of skills, all being stored in $k$ different tables. The join query would now be
between $k$ tables and it would give a naive linear program with $O(N^k)$
variables. Our technique would still construct an equivalent linear program
having $O(N)$ variables, leading to more interesting improvements. 

We now illustrate how one can rewrite the original linear program into a smaller
one. We introduce the following variables instead of the assignment weights. We
denote by $\exW{P}{pfl}$ the time invested into the project $p$ by pairs of
researchers and programmers in field $f$ and language $l$. We denote by
$\exW{R}{rf}$ the percentage of time invested by researcher $r$ into projects in
the field $f$. Similarly we define by $\exW{D}{dl}$ the amount of time invested
by developer $d$ into projects requiring the language $l$.

Observe that, given $f$ and $l$, $\exW{P}{pfl}$ corresponds to $\sum\limits_{r,d} \exW{P}{prdfl}$ 
in the previous model. 
The objective function of the LP can thus be simplified through a change of variables.
Given a field $f$, we want to enforce the total time invested into all the projects concerning 
field $f$ to be equal to the total time invested by researchers into projects in field $f$ 
which can be formally expressed as $\forall f \in \domain{\sciExpertise}, 
            \sum\limits_{r \in \sciContributor} \exW{R}{rf} =
            \sum\limits_{p \in \project} \sum\limits_{l \in \expExpertise} \exW{P}{pfl}$.
Similarly we will express a constraint on the projects in a given language $l$ and developers
working with language $l$. The "pairing" constraint requiring that projects see the same amount
of work from researchers and developers is implicitly enforced by the transitivity of the equality.
Finally the constraints imposing an upper bound on the workload of the researchers and developers
can easily be translated with these new weights. The complete rewriting of the linear program is
shown in Fig.~\ref{fig:lprewritten}.

\begin{figure}

\begin{align*}
    \text{maximize } & \sum\limits_{(p, f, l) \in \project} \exW{P}{pfl} \\
    \text{subject to }  
        & \forall f \in \domain{\sciExpertise}, 
            \sum\limits_{r \in \sciContributor} \exW{R}{rf} =
            \sum\limits_{p \in \project} \sum\limits_{l \in \expExpertise} \exW{P}{pfl} \\
        & \forall l \in \domain{\expExpertise}, 
            \sum\limits_{d \in \expContributor} \exW{D}{dl} =
            \sum\limits_{p \in \project} \sum\limits_{f \in \sciExpertise} \exW{P}{pfl} \\
        & \forall r \in \domain{\sciContributor}, \sum\limits_{f \in \domain{\sciExpertise}} \exW{R}{rf} \leq 100\\
        & \forall d \in \domain{\expContributor}, \sum\limits_{l \in \domain{\expExpertise}} \exW{D}{dl} \leq 100\\
\end{align*}

\caption{Rewriting of the lp example}
\label{fig:lprewritten}
\end{figure}

Observe that we introduce a single variable for each tuple of the input tables
and that each of the four constraints above can be unfolded in a number of
inequalities that is bounded by the cardinality of the domains of the attributes
of the input tables. Thus, exploiting the acyclicity of $Q$, we transformed a
linear program having $O(N^3)$ variables to a linear program having $O(N)$
variables and the same optimal value. Our main result is an algorithm to
automatically provide similar rewrittings for a fragment of linear programs on
tables.



\subsection{Contribution of this paper}
\label{sec:contrib}

In this paper, we generalize the technique seen in the previous example to
reduce the size of a linear program whose variables are the answer set of a
conjunctive query. More precisely, we give a generic approach to exploit the
structure of conjunctive queries in order to generate equivalent linear programs
of size polynomial in the size of the database which may be smaller than the
size of the answer set. As deciding whether a conjunctive query has at least
one solution is already $\NP$-complete, it is very unlikely that this approach
is always possible. Hence, we based our approach on the framework of factorized
databases of Olteanu and Z\'{a}vodn\'{y}~\cite{olteanu2012factorised}. A
factorized database is a data structure that allows to represent a table in a
factorized way that can be much smaller than the original table. More precisely,
a factorized representation of a table may be seen as a circuit whose inputs are
tables with one column and performing only Cartesian products and disjoint
unions that we will refer to as {\ddcircuit}s in this paper. Our central
theorem, Theorem~\ref{th:weightingCorrespondence}, shows that a given weighting
of a table represented by a {\ddcircuit} can be canonically represented as a
weighting of the edges of the circuit. This connection gives us a direct way of
rewriting a linear program whose variables are the tuples of a table represented
by a {\ddcircuit} into a linear program whose variables are the edges of the
circuit having the same optimal value.

Several classes of conjunctive queries are known to admit polynomial size
{\ddcircuit}s, the most notorious one being conjunctive queries with bounded
fractional hypertree width~\cite{olteanu_size_2015}. Fractional hypertree
width~\cite{gottlob2016hypertree,gottlob_hypertree_2002} is a parameter
generalizing acyclicity of the hypergraph of a conjunctive
query~\cite{Yannakakis}. Intuitively, it measures how the relations and the
variables of a conjunctive query interact with each other. Various aggregation
problems are tractable on such queries such as counting the number of
answers~\cite{pichler_tractable_2013} or enumerating them with small
delay~\cite{deep2019ranked,bagan2007acyclic}. Most of these problems are solved
by dynamic programming algorithms working in a bottom-up fashion on a tree
shaped decomposition of the query. While one could use this dynamic programming
approach to solve our problem as well, it makes the algorithm harder to
describe. We thus choose to present our results on {\ddcircuit}s directly. In
this setting, given a {\ddcircuit} and a linear program on its answer set, the
rewriting of the linear program is straightforward, even if the correction of
the algorithm is more complicated to establish.

\subsection{Related work}
\label{sec:related}

Factorized databases have already been used to solve complex problems on the
answers of conjunctive queries. They have been used by Schleich et
al.~\cite{schleich2016learning} to perform linear regression on data represented
as a factorized databases more efficiently than by first enumerating the
complete answer set. Fink et al.~\cite{fink2012aggregation} have also used
factorized databases to compute the probability distribution of query results in
the framework of probabilistic databases. To the best of our knowledge, our
result is the first to use this framework to solve optimization problems more
efficiently.

Better integration of optimization problems directly into DBMS has however
already been considered. Cadoli and Mancini~\cite{cadoli07} introduced an
extension of \texttt{SQL} called \texttt{ConSQL} that allows to select optimal
solutions to constraint problems directly in \texttt{SQL}-like queries.
\v{S}ik{\v{s}}nys and Pedersen introduced \texttt{SolveDB} in~\cite{vsikvsnys16}
which is also an extension of \texttt{SQL} that allows to solve optimization
problems directly in the querying language thus simplifying the usual
workflow of feeding an \texttt{AMPL} programs with data extracted from a
database.
Both work mainly focus on designing the language and procedures to
answer query in a very general setting. Though this is a very important and
interesting question, we do not focus much on the language part in this work but
more on the optimization part of solving linear programs.

The need to solve linear programs on answer sets of conjunctive queries has
already appeared in the data mining community, especially for the problem of
estimating the frequency of a pattern in a data graph. In this setting, the
problem is to evaluate the frequency of a subgraph, the pattern, in a larger
graph. A naive way of evaluating this frequency is to use the number of
occurrences of the pattern as a frequency measure. Using this value as a
frequency measure, is problematic since different occurrences of a pattern may
overlap, and as such they share some kind of dependencies that is relevant from
a statistical point of view. More importantly, due to the overlaps, this measure
fails to be anti-monotone, meaning that a subpattern may be counter-intuitively
matched less frequently than the pattern itself. Therefore, the finding of
better anti-monotonic frequency measures -- also known as {\em support measures}
-- has received a lot of attention in the data mining
community~\cite{bringmann08,calders11,fiedler07}. A first idea is to count the
maximal number of non-overlapping patterns~\cite{vanetik02}. However, finding
such a maximal subset of patterns essentially boils down to finding a maximal
independent set in a graph, a notorious NP-complete problem~\cite{garey02}. The
$s$-measure, an alternative approach proposed by Wang et
al.~\cite{wang_efficiently_2013}, consists in counting the number of occurrences
with respect to a weighting of the occurrences taking dependencies into account.
The weighting proposed by the authors is defined via a linear program whose
variables are the occurrences of the pattern. Seeing the problem of finding all
occurrences of a pattern in a graph as a conjunctive query on a database whose
only relation is the edges of the graph, our result directly applies to the
problem of computing the $s$-measure and allows to compute it in polynomial time
if the treewidth of the pattern is bounded.

Binary and Mixed-Integer programs (which are closely related to linear programs) have also 
appeared in the database community as ways to work with possible worlds. In 
\cite{kolaitis_efficient_2013} the authors use binary programming to answer queries on 
inconsistent databases. In this work a binary programs concisely encodes the possible
repairs of the database and is used to eliminate potential answers until a consistent answer
(or the absence thereof) is found. In \cite{tiresias} the authors use Mixed-Integer Programming 
to answer how-to queries expressed in the "Tiresias Query Language" (TiQL) which is based on datalog.
There a Mixed-Integer Program is generated from the TiQL query and database. The MIP is then
optimized using the structure of the database before being fed to an external solver which is 
similar in spirit to the goal of this paper.

\paragraph{Organization of the paper.}
The paper is organized as follows. Section~\ref{sec:preliminaries} introduces
the notions and definitions that are necessary to precisely state our main
result. In particular, this section contains the definition of the fragments of
linear programs on tables we use in this paper. Section~\ref{sec:weight-corr}
presents our main technical tool to solve the general problem by relating
weights given on the rows of a table and on the edges of the circuit
representing it in a factorized manner and explains how a linear program whose
variables are the elements of a relation defined by a circuit can be rewritten
as a linear program whose variables are the edges of the circuit and thus be
solved in polynomial time in the size of the circuit. The correction of our
algorithm is presented in Section~\ref{sec:proofs}. Section~\ref{sec:proj}
exposes some consequences of our results: we mention corollaries to prove the
tractability of bounded fractional hypertreewidth quantifier-free conjunctive
queries. We also present a simple method to handle existential quantification.
Finally, we give suggestions for future work in the conclusion.

\section{Preliminaries}
\label{sec:preliminaries}
\subsection{General notations}

We denote by $\R$ the set of real numbers and by $\R_+$ the set of non-negative
real numbers. For any positive natural number $n$, we denote by $[n]$ the set
$\{1,\dots,n\}$. For two sets $Y$ and $Z$, we denote by $Y^Z$ the set of (total)
functions from $Z$ to $Y$. A finite function
$R=\{(z_1,y_1),\ldots,(z_n,y_n)\}$ where all $z_i$ are pairwise distinct is
denoted by $[z_1/y_1,\ldots,z_n/y_n]$. In the case $n=1$ we
will simply write $z_1/y_1$ instead of $[z_1/y_1]$. Given a function $f \in Y^Z$ and $Z' \subseteq Z$, we
denote by $\projt{f}{Z'}$ the projection of $f$ on $Z'$, that is,
$\projt{f}{Z'} \in Y^{Z'}$ and for every $z \in Z'$,
$\projt{f}{Z'}(z)=f(z)$

A directed graph is a pair $G=(V,E)$ with $E\subseteq V\times V$.
Given a directed graph $G$ we write $\nodes{G}=V$ for the set of
nodes and $\edges{G}=E$ for the set of edges. For any $u\in V$, 
we denote by $\ingoing{u}$ (resp. $\outgoing{u}$) the set of 
ingoing (resp. outgoing) edges of $u$.

Let $\varx$ be a finite set of \emph{attributes} and $\dom$ be a finite domain.
A \emph{(database) tuple with attributes in $\varx$ and domain $\dom$} is a
function $\tau:\varx\to \dom$ mapping attributes to database elements. A
\emph{(database) relation} with attributes $\varx$ and domain
$\dom$ is a set of such tuples, that is, a subset of $\dom^\varx$.


\subsection{Linear programs}
\label{sec:lp}

Linear programs form a wide class of (convex) optimization problems which are
known to be tractable. The first polynomial time algorithm for solving this
problem is the ellipsoid algorithm~\cite{khachiyan1979}. Other more efficient
algorithms have been proposed later~\cite{karmarkar, lee} and there exist many
tools using these techniques together with heuristics to solve this problem in
practice such as {lp\_solve\footnote{\url{http://lpsolve.sourceforge.net/}}} or
{glpk\footnote{\url{https://www.gnu.org/software/glpk/}}}. In this paper, we
will use these tools as black boxes and will only use the fact that finding
optimal solutions of linear programs can be done in polynomial time.

Given a finite set of variables $X$, a \emph{linear form} $\phi$ on $X$ is given
as a set of reals $\mu(x)$ for every $x \in X$. Given an assignment $a : X
\rightarrow \R$ of the variables of $X$, $\phi(a)$ is defined as $\sum_{x \in X}
\mu(x)a(x)$. Through the paper, we denote linear forms as polynomials as
follows:
\[\mu_1 X_1 + \dots + \mu_n X_n. \]

A \emph{linear constraint} $C=(\phi,s,c)$ on variables $X$ is given as a linear
form $\phi$ on $X$, a binary relation $s \in \{\leq, =, \geq, <, >\}$
and a real $c \in \R$. Given an assignment $a : X \rightarrow \R$ of the
variables of $X$, $C$ is satisfied by $a$ if $s(\phi(a), c)$. Through the
paper, we denote linear constraints by adapting the notation of linear programs:
\[\mu_1 X_1 + \dots + \mu_n X_n \leq c. \]

A \emph{linear program} $L=(X, \phi, \calC)$ on variables $X$ is given as a
linear form $\phi$ called the \emph{objective function} and a set of linear
constraints $\calC$ on variables $X$. We write $L$ as:
$$
\begin{gathered}
    \text{max } \phi  \text{ subject to } \calC
\end{gathered}
$$

The \emph{feasible region} of $L$, denoted by $\sol{L}$, is the set of assignments 
$a : X \rightarrow \R$ such that for every $C \in \calC$, $a$ satisfies $C$. We will often refer to
an assignment in the feasible region as a \emph{solution} to the linear program.
An \emph{optimal solution} of $L$ is a solution $a$ in the feasible region that
maximizes the objective function of $L$, i.e. for all $a' \in \sol{L}$, $\phi(a)
\geq \phi(a')$. The \emph{optimal value} of $L$ is defined as $\phi(a)$ where
$a$ is some optimal solution of $L$.

The \emph{size} of $L$ is defined to be $|X| \cdot |\calC|$. Observe that the
size of an encoding of a linear program may be bigger than our definition of
size as we will have to store the coefficient. In this paper, we present a
succinct rewriting of a class of linear program that does not change the
coefficients, so we do not need to be too precise on the way they are encoded.
In order to simplify the presentation of the results, we will not take this
encoding into account.

\subsection{Linear programs on relations}
\label{sec:lpt}

In this paper, we are interested in a class of linear programs whose variables
are the elements of a finite relation. To manipulate this kind of programs, we
need a definition of linear programs parametrized by a relation. We introduce in
this section the notion of \emph{Common Attribute Sum Linear Programs},
abbreviated as {\sumlp}s. A {\sumlp} $L=(\attrset, \dom, \phi, \calC)$ on
attributes $\attrset$ and domain $\dom$ is given as a linear form $\phi$ called
the \emph{objective function} and a set of linear constraints, both on variables
$S(\attrset, \dom) := \{S_{x,d}\}_{x \in \attrset, d \in \dom}$, called the
{\sxdname}s.

From a strictly formal point of view, a {\sumlp} is just a linear program on
variables $S(\attrset,\dom)$. However, we are not interested in the optimal
solutions of this linear program but in the solution of a linear program that is
obtained from a {\sumlp} and a relation $R \subseteq \dom^\attrset$.

Given a {\sumlp} $L$ on attributes $\attrset$, domain $\dom$ and given a
relation $R \subseteq \dom^\attrset$, a solution of $L$ on $R$ is an assignment
$\om: R \rightarrow \R_+$ such that the extension of $\om$ on variables
$S_{x,d}$ defined as $\om(S_{x,d}) := \xdtuplesum{x}{d}{R}$ is a solution of
$L$. An \emph{optimal solution} of $L$ on $R$ is a solution of $L$ on $R$ that
maximizes the objective function of $L$. The \emph{optimal value} of $L$ on $R$
is the value of the objective function on an \emph{optimal solution}.

Alternatively, we could see a solution of $L$ on $R$ as the solution of the
\emph{ground linear programs} $L(R)$ on variables $\{\tau \mid \tau \in R \}$
obtained by replacing the variable $S_{x,d}$ by $\sum\limits_{\substack{\tau \in
    R \\ \tau(x)=d}} \tau$ in $L$ and by adding the constraints $\tau \geq 0$
for every $\tau \in R$. This gives a first naive algorithm to find the optimal
solution of $L$ on $R$: one can compute $L(R)$ explicitly and solve it using an
external solver.

When $R$ is given as the relation defined by a conjunctive query on a database
however, the relation may be much bigger than the size of the database which
results in a huge ground linear program. The goal of this paper is to study how
one can compute the optimal value of such ground linear programs in time
polynomial in the size of the database and not in the size of the resulting
relation.

Fig.~\ref{fig:lpsxd} gives a {\sumlp} $L$ corresponding to the example of
Section~\ref{sec:example}. Indeed, if $R$ is the relation corresponding to the
answer set of query $Q$ of Section~\ref{sec:example}, then $L(R)$ exactly
corresponds to the linear program of Fig.~\ref{fig:lpex}.

\begin{figure}

\begin{align*}
    \text{maximize } & \sum\limits_{v \in \project} \sxd{p}{v} \\
    \text{subject to }  
        & \forall v \in \domain{\sciContributor}, 
            \sxd{r}{v} \leq 100 \\
        & \forall v \in \domain{\expContributor}, 
            \sxd{d}{v} \leq 100 \\
\end{align*}

\caption{Research projects example as a \sumlp}
\label{fig:lpsxd}
\end{figure}

Actually, computing this optimal value for conjunctive query is $\NP$-hard:

\begin{proposition}
  The problem of deciding whether the optimal value of $L(Q(\db))$ is non-zero
  given a quantifier-free conjunctive query $Q$, a database $\db$ and a {\sumlp}
  $L$ is $\NP$-hard.
\end{proposition}
\begin{proof}
  Let $Q$ be a quantifier-free conjunctive query and $x$ an attribute of $Q$ on
  domain $\dom$. Let $L$ be the {\sumlp} defined as:
\begin{align*}
    \text{maximize } & \sum\limits_{d \in \dom} \sxd{x}{d}.
\end{align*}
It is not hard to see that, given a database $\db$, the optimal value of
$L(Q(\db))$ is non-zero if and only if $Q(\db) \neq \emptyset$. Since deciding,
given $Q$ and $\db$, whether $Q(\db) \neq \emptyset$ is an $\NP$-complete
problem~\cite{chandra77}, the statement of the proposition follows.
\end{proof}

\subsection{{\ddCircuit}s}
\label{sec:cq}

To solve ground linear programs efficiently, we will work on factorized
representations of relations. We use a generalization of the framework of
factorized databases introduced by Olteanu et al.~\cite{olteanu2012factorised}.
Our factorized representations are defined similarly with only a subset of the
syntactic restrictions (we do not need $f$-representation or
$d$-representations) on the circuits since our algorithm does not need all of
them to work.

A {\em $\{\cup,\times\}$-circuit} on attributes $\varx$ and domain $\dom$ is a
directed acyclic graph $\circuit$ with edges directed from the leaves to a
single root $r$ called the output of $\circuit$ and whose nodes, called
\emph{gates}, are labeled as follows:
\begin{itemize}  
    \item internal gates are labeled by either $\times$ or $\cup$,
    \item inputs, \emph{i.e.} nodes of in-degree $0$, are labeled by
      atomic database relations of the form $x/d$ for an attribute
        $x$ and an element $d \in \dom$. \footnote{Observe that every possible input
      relation could be easily simulated with this restricted input using
      $\times$ and $\cup$. For example, the relation $\{[x/1, y/1, z/1], [x/0,
      y/0, z/0] \}$ can be rewritten as $(x/1 \times y/1 \times z/1) \cup (x/0
      \times y/0 \times z/0)$.}
\end{itemize}

Figure~\ref{fig:circuit} pictures a $\{\cup, \times\}$-circuit on attributes
$\{x,y,z\}$ and domain $\{0,1\}$. 
\begin{figure}
  \centering
  \begin{minipage}{.3\textwidth}
      \includegraphics[width=4cm]{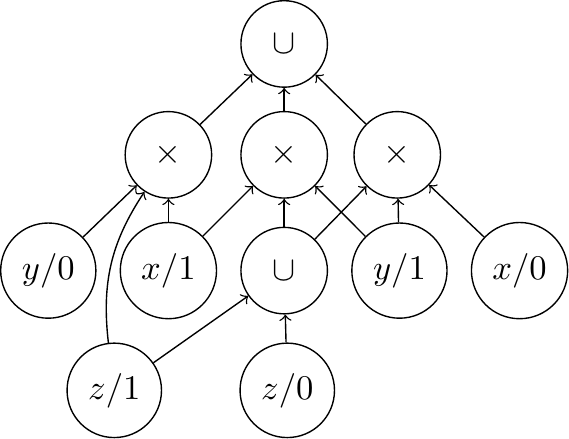}
  \end{minipage}%
  \begin{minipage}{.1\textwidth}
    \begin{tabular}{|c|c|c|}\hline
      $x$ & $y$ & $z$   \\ \hline\hline
      $1$ & $1$ & $0$       \\ \hline
      $1$ & $0$ & $1$       \\ \hline
      $0$ & $1$ & $1$       \\ \hline
      $0$ & $1$ & $0$       \\ \hline
      $1$ & $1$ & $1$       \\ \hline
    \end{tabular}
\end{minipage}
  \caption{Example of a $\{\cup,\times\}$-circuit representing the
   database relation on the right. This circuit is also a \ddcircuit.}
    \label{fig:circuit}   
\end{figure}
We denote by $\var{\circuit}$ the attributes appearing in the database
relations of the input nodes  of
$\circuit$. 
We define the \emph{size} of $\circuit$ as the number of edges in the
underlying graph of $\circuit$. We denote it by $|C|$.

Given a gate $u$ of $\circuit$, we denote by $\circuit_u$ the
subcircuit of $\circuit$ rooted in $u$ (in which
dangling edges from above are removed). We denote by 
{\andgate}s$(\circuit)$ the set of all the {\andgate}s
and by \cupgate{}s$(\circuit)$ the set of all
\cupgate{}s of $\circuit$.

We impose that any $\{\cup, \times\}$-circuit must
  satisfy the following restrictions.
\begin{itemize}  
\item If $u$ is a $\cup$-gate with children $u_1,\dots,u_k$, we impose that
$\var{\circuit_{u_1}} = \dots = \var{\circuit_{u_k}}$.
\item If $u$ is a $\times$-gate with children $u_1, \dots, u_k$, we impose
  that $\var{\circuit_{u_i}} \cap \var{\circuit_{u_j}} = \emptyset$ for every $i
  < j \leq k$. Observe that it implies that $\circuit_{u_i}$ and
  $\circuit_{u_j}$ are disjoint graphs since otherwise they would share an input
  and, thus, an attribute of the relation of this input node.
\end{itemize}
Due to these restrictions, any gate $u$ of a $\{\cup,
  \times\}$-circuit $C$ specifies a database relation $\solcircuit{u}$ 
that we define inductively as follows. If $u$ is an input then $\solcircuit{u}$ is the
relation labelling $u$. If $u$ is a $\times$-gate with children $u_1, \dots,
u_k$, $\solcircuit{{u}}$ is defined as $\solcircuit{u_1} \times \dots \times
\solcircuit{u_k}$. If $u$ is a $\cup$-gate with children $u_1, \dots, u_k$,
$\solcircuit{{u}}$ is defined as $\solcircuit{u_1} \cup \dots \cup
\solcircuit{u_k}$. We define $\solcircuit{}=\solcircuit{r}$ where $r$ is the output of the
circuit. 




The way we have currently defined circuits is not yet enough to allow for
efficient aggregation. Given a $\{\cup,\times\}$-circuit $\circuit$ and an
$\cup$-gate $u$ of $\circuit$ with children $u_1,\dots,u_k$, we say that $u$ is
\emph{disjoint} if $\solcircuit{{u_i}} \cap \solcircuit{{u_j}} = \emptyset$ for
every $i < j \leq k$.
  In our figures and proofs, we will use the symbol $\uplus$ to indicate
  that a union is disjoint. A {\ddcircuit} is a $\{\cup,\times\}$-circuit in
  which every $\cup$-gate is a \orgate{}. 

 Observe that given a {\ddcircuit}
$\circuit$, one can compute the cardinality of its relation $|\solcircuit{}|$ in time linear in the size of
$\circuit$ by a simple inductive algorithm that adds the sizes of the children
of a $\uplus$-gate and multiply the sizes of the children of the $\times$-gates.

Disjointness is a semantic condition and it is $\coNP$-hard to check the
disjointness of a gate given a \udcircuit. In practice however, algorithms
producing {\ddcircuit} ensure the disjointess of all $\cup$-gates during the construction of the
circuit, usually by using $\cup$-gates of the form: $\biguplus_{d \in D}
(\{x/d\} \times C_d)$, so that the different values $d$ assigned to $x$
ensure disjointness.

\begin{theorem}{\cite[Theorem 7.1]{olteanu_size_2015}}
  \label{thm:cq-to-fr} Given a quantifier free conjunctive query $Q$, a
  hypertree decomposition $T$ of $Q$ of fractional hypertreewidth $k$ and a
  database $\db$, one can construct a \ddcircuit $\circuit$ of size at most $|Q|
\cdot |\db|^k$ whose attributes are the variables of $Q$ such that $Q(\db) =
  \solcircuit{}$ in time $O(\log(|\db|)\cdot |\db|^k)$.
\end{theorem}

For the special case of quantifier free acyclic conjunctive query the
theorem follows from a variant of Yanakakis' algorithm. 
We refer the reader to~\cite{gottlob2016hypertree} for the definition of
fractional hypertreewidth as we will not work with this notion directly in this
paper but only with {\ddcircuit}s, using Theorem~\ref{thm:cq-to-fr} to bridge
both notions.

\section{Rewriting linear programs}
\label{sec:weight-corr}

The goal of this paper is to show that one can rewrite the ground linear program
$L(R)$ of a {\sumlp} $L$ when the relation $R$ is structured. More precisely, we
show the following theorem:

\begin{theorem}
  \label{th:rewritecaslp} Given a {\sumlp} $L$ with $m$ constraints and a
  \ddcircuit $C$ both on attributes $X$ and domain $D$, one can construct in
  polynomial time a linear program $L'$ having $|C|$ variables and $m+O(|C|)$
  constraints such that $L'$ and $L(\solcircuit{})$ have the same optimal value.
\end{theorem}

\subsection{Construction of the smaller linear program}
\label{sec:smalllp}

We start by presenting the construction of the new linear program of
Theorem~\ref{th:rewritecaslp}. We fix $L$ a {\sumlp} with $m$ constraints and
$C$ a \ddcircuit, both on attributes $X$ and domain $D$. We construct a linear
program $L'$ having the same optimal value as $L(\solcircuit{})$. The variables
of $L'$ are the edges of $\circuit$. $L'$ has a first set of constraints that we
call the \emph{soundness constraints} that are the following. For every edge $e$
of $\circuit$, we have a constraint $e \geq 0$ and for every gate $u$ of $C$,
\begin{itemize}
\item if $u$ is a \orgate{} that is not the output of the circuit then we have
  the constraint : $\sum_{i \in \ingoing{u}} i = \sum_{o \in \outgoing{u}} o$,
\item if $u$ is a \andgate{} then let $i_1, \dots, i_k$ be the ingoing edges of
  $u$ in an arbitrary order. We have the constraint $i_p = i_{p+1}$ for $1 \leq
  p < k$. If $u$ is not the output of the circuit, we also have the constraint
  $i_1 = \sum_{o \in \outgoing{u}} o$.
\end{itemize}
Observe that there are at most $3|C|$ soundness constraints in $L'$. Now, we add
a constraint $c'$ for each constraint $c$ of $L$ that is obtained as follows: we
replace every occurrence of $S_{x,d}$ in $c$ by $\sum_{e \in \outgoing{x/d}} e$
for every $x \in X$ and $d \in D$ where $\outgoing{x/d}$ denotes the set of
edges of $C$ going out of an input labeled by $x/d$. The objective function of
$L'$ is obtained by applying the same replacement of $S_{x,d}$ variables in the
objective function of $L$.

The number of variables of $L'$ is $|C|$ and that it has $m+3|C|$ constraints so
it has the required size. It is also clear that one can easily construct $L'$
from $L$ and $C$ in polynomial time: it is indeed sufficient to visit each gate
of the circuit to generate the soundness constraints and then to replace each
occurrences of $S_{x,d}$ by the appropriate sum.

It turns out that the optimal values of $L'$ and of $L(\solcircuit{})$ coincide.
To prove this however, we have to understand the structure of {\ddcircuit}
better and the rest of the paper is dedicated to prove the correction of our
construction.

\subsection{Correctness of the construction}

A solution of $L(\solcircuit{})$ can be seen as a weighting of the tuples $\tau
\in \solcircuit{}$. Similarly, a solution of $L'$ can be seen as a weighting of
the edges of $\circuit$. We will show strong connections between both weightings
that will allow us to prove that $L(\solcircuit{})$ and $L'$ have the same optimal
value.

Let $\circuit$ be a \ddcircuit. A \emph{tuple-weighting} $\om$ of $\circuit$ is
defined as a variable assignment of $\solcircuit{}$, e.g. $\om \in
\tweightings$. An \emph{edge-weighting} $\W$ of $\circuit$ is similarly defined
as a variable assignment of $\edges{\circuit}$, e.g. $\W \in \eweightings$. 
Throughout the paper, we will always use the symbol $\om$ for tuple-weightings
and $\W$ for edge-weightings.

\begin{definition}
An edge-weighting $\W$ is \emph{sound} if for every gate $u$ of $\circuit$ we have:
\begin{itemize}
    \item if $u$ is a \orgate{} that is not the output of the circuit then
        $\sum_{i \in \ingoing{u}} W(i) = \sum_{o \in \outgoing{u}} W(o)$,    
    \item if $u$ is a \andgate{} then $\forall i,i' \in \ingoing{u}$:
         $\val{\W}{i} = \val{\W}{i'}$. If $u$ is not the output of the circuit,
         we also have $W(i) = \sum_{o \in \outgoing{u}} W(o)$.
       \end{itemize}
\end{definition}

Given a tuple-weighting $\om$ and a tuple-weighting $\W$ we say that $\W$ is
compatible with $\om$ iff $\forall x \in X, \forall d \in \dom 
\xdtuplesum{x}{d}{\solcircuit{}} = \sum\limits_{e \in \outgoing{x/d}} \val{\W}{e}$.

The correctness of our construction directly follows from the next theorem whose
proof is delayed to Section~\ref{sec:proofs}, connecting tuple-weightings with
sound edge-weightings.

\begin{theorem}
    \label{th:weightingCompatibility}
    Given a {\ddcircuit} $\circuit$. 
    \begin{enumerate}[(a)]
    \item For every tuple-weighting $\om$ of $\circuit$, there exists a sound
      edge-weighting $\W$ such that $\W$ is compatible with $\om$.
    \item Given a sound edge-weighting $\W$ of $\circuit$, there exists a
      tuple-weighting $\om$ of $\circuit$ such that $W$ is compatible with $\om$. 
    \end{enumerate}    
\end{theorem}

We now explain how one can use Theorem~\ref{th:weightingCompatibility} to prove
the correctness of our transformation.

Let $\om: \solcircuit{} \rightarrow \R_+$ be a tuple-weighting that is a
solution of $L(\solcircuit{})$. Let $\W : \edges{\circuit} \rightarrow \R_+$ be
the sound edge-weighting compatible with $\om$ that is given by
Theorem~\ref{th:weightingCompatibility}. We claim that $W$ is a solution of
$L'$. First, $W$ satisfies all the soundness constraints of $L'$ since $W$ is
sound. Moreover, since $W$ is compatible with $\tau$, $\forall x \in X, \forall
d \in \dom \xdtuplesum{x}{d}{\solcircuit{}} = \sum\limits_{e \in \outgoing{x/d}}
\val{\W}{e}$. But this is the value of $\om(S_{x,d})$, which means that $W$ is a
solution of $L'$. It tells us that the optimal value of $L'$ is bigger than the
optimal value of $L$.

Now let $\W : \edges{\circuit} \rightarrow \R_+$ be an edge-weighting of
$\circuit$ that is a solution of $L'$ and let $\om$ be a tuple-weighting
compatible with $W$. We show that $\om$ is a solution of $L(\solcircuit{})$.
Indeed, by definition, $L'$ is obtained by replacing $S_{x,d}$ with
$\sum\limits_{e \in \outgoing{x/d}} e$. Since $\om$ is compatible with $\W$, we
also have $\forall x \in X, \forall d \in \dom \xdtuplesum{x}{d}{\solcircuit{}}
= \sum\limits_{e \in \outgoing{x/d}} \val{\W}{e}$. But this is $\om(S_{x,d})$ by
definition, which means that $\om$ is a solution of $L(\solcircuit{})$. Thus the
optimal value of $L(\solcircuit{})$ is bigger than the optimal value of $L'$.
Hence, $L(\solcircuit{})$ and $L'$ have the same optimal value.

The rest of this paper is dedicated to the proof of
Theorem~\ref{th:weightingCompatibility}.

\section{Connecting edge-weightings and tuple-weightings}
\label{sec:proofs}

In this section, we prove Theorem~\ref{th:weightingCompatibility}. Actually, we
prove a stronger version of this theorem presented in
Section~\ref{sec:th-statement} that gives a better insight into how
tuple-weighting and edge-weighting relate to each others. This stronger version
allows us to prove the result by induction on the circuit. Before stating the
theorem, we however need to introduce some notions on circuits.

\subsection{Proof-trees}
\label{sec:proof-trees}

From now on, to simplify the proofs, we assume wlog that every internal gate of
a {\ddcircuit} has fan-in two. It is easy to see that, by associativity,
$\times$-gates and $\uplus$-gates of fan-in $k>2$ can be rewritten with $k-1$
similar gates which is a polynomial size transformation of the circuit. We also
assume that for every $x \in \var{\circuit}$ and $d \in D$, we have at most one
input labeled with $x/d$. This can be easily achieved by merging all such
inputs.

Let $\circuit$ be a {\ddcircuit} and let $\tau \in \solcircuit{}$. The
\emph{proof-tree} of $\tau$, denoted $\prooftree{\circuit}{\tau}$, is a
subcircuit of $\circuit$ participating to the computation of $\tau$. More
formally, $\prooftree{\circuit}{\tau}$ is defined inductively by starting from
the output as follows: the output of $\circuit$ is in
$\prooftree{\circuit}{\tau}$. Now if $u$ is a gate in
$\prooftree{\circuit}{\tau}$ and $v$ is a child of $u$, then we add $v$ in
$\prooftree{\circuit}{\tau}$ if and only if $\proj{C_v} \in \solcircuit{v}$. A
proof-tree is depicted in red in Figure~\ref{fig:prooftree_example}.

\begin{figure}
  \centering
    \includegraphics[width=4cm]{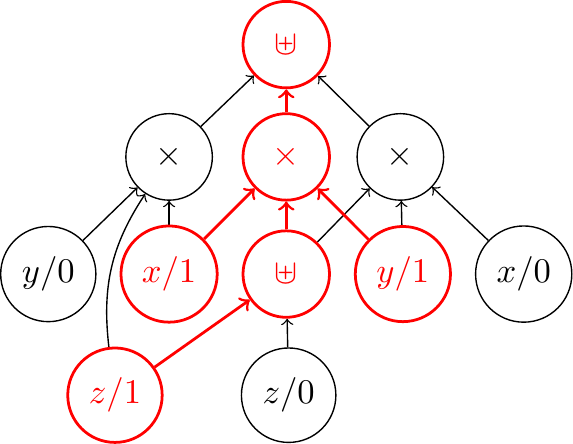}
    \caption{Proof-tree of the tuple $[x/1, y/1, z/1]$}
    \label{fig:prooftree_example}
\end{figure}

\begin{proposition}
    \label{prop:prooftreeUnique}
    
    Given a \ddcircuit{} $\circuit$ and $\tau \in \solcircuit{}$,
    $\somept := \prooftree{\circuit}{\tau}$, the following holds:
    \begin{itemize}
    \item every $\times$-gate of {$\somept$} has all its
      children in $\somept$,
    \item every $\uplus$-gate of $\somept$ has exactly one of
      its children in $\somept$,      
    \item $\somept$ is connected and every gate of $\somept$ has out-degree at
      most $1$ in $\somept$,
    \item for any $x \in \var{\circuit}$, $\somept$ contains exactly one input 
    labeled with $x/\tau(x)$.
    \end{itemize}
\end{proposition}

\begin{proof}
    \label{prf:prooftreeUnique}
    By induction, it is clear that for every gate $u$ of $\somept$, $\proj{C_u}
    \in \solcircuit{u}$. Thus, if $u$ is a $\uplus$-gate of $\somept$, as $u$ is
    disjoint, exactly one of it child $v$ has $\proj{C_u} \in \solcircuit{u}$.
    If $u$ is a $\times$-gate with children $u_1,u_2$, then by definition,
    $\proj{C_u} \in \solcircuit{u}$ if and only if $\proj{C_{u_1}} \in
    \solcircuit{u_1}$ and $\proj{C_{u_2}} \in \solcircuit{u_2}$. Thus both $u_1$
    and $u_2$ are in $\somept$.

    It is clear from definition that $\somept$ is connected since $\somept$ is
    constructed by inductively adding children of gates that are already in
    $\somept$. Now assume that $\somept$ has a gate $u$ of out-degree greater
    than $1$ in $\somept$. Let $u_1,u_2$ be two of its parents and let $v$ be
    their least common ancestor in $\somept$. By definition, $v$ has in-degree
    $2$, this it is necessarily a $\times$-gates with children $v_1,v_2$. Thus
    $u$ is both in $\circuit_{v_1}$ and $\circuit_{v_2}$, which is impossible
    since they are disjoint subcircuits.

    Finally, let $x \in \var{\circuit}$. Observe that if $\somept$ has an input
    labeled $u$ with $x/d$ then $d = \tau(x)$ since $\proj{C_u} \in
    \solcircuit{u}$. Thus, if $\somept$ contains two inputs $v_1$ and $v_2$ on
    attribute $x$, they are both labeled with $x/\tau(x)$ and are thus the same
    input.    
\end{proof}

A proof-tree may be seen as the only witness that a tuple belongs to $\solcircuit{}$.
We define the relation induced by an edge $e$
as the set of tuples of $\solcircuit{}$ such that their proof-tree contains
the edge $e$, that is $\soledge{e} := \{ \tau \in \solcircuit{} \mid e \in
\prooftree{\circuit}{\tau} \}$.

In the following we show some fundamental properties of proof-trees that we will
use to prove the correctness of our algorithm. For the rest of this section, we
fix a {\ddcircuit} $\circuit$ on attributes $\varx$ and domain $\dom$.

\begin{proposition}
  \label{prop:outgoingDisjoint}
  Let $u$ be a gate of $\circuit$ and $e_1, e_2$ $\in \outgoing{u}$
    with $e_1 \neq e_2$.
  Then $\soledge{e_1} \cap \soledge{e_2} = \emptyset$.
\end{proposition}
\begin{proof}
  Assume that $\tau \in \soledge{e_1}$ and $\tau \in \soledge{e_2}$ for $e_1,e_2
  \in \outgoing{u}$ with $e_1 \neq e_2$. By definition, it means that both $e_1$ and $e_2$ are in
  $\prooftree{\circuit}{\tau}$ which is a contradiction as
  $\prooftree{\circuit}{\tau}$ has out-degree $1$ by
  Proposition~\ref{prop:prooftreeUnique}.
\end{proof}

\begin{proposition}
    \label{prop:prooftreeInput}    
    Let $x \in \varx$ and $d \in \dom$ and let $u_{x/d}$ be the input of $\circuit$
    labeled with $x/d$. 
    \[ \tupleset{x}{d} = \biguplus\limits_{e \in \outgoing{u_{x/d}}} \soledge{e}. \]
\end{proposition}
\begin{proof}
  Let $x \in \var{\circuit}$ and $d \in \dom$. Let $\tau \in \tupleset{x}{d}$.
  By Proposition \ref{prop:prooftreeUnique}, $u_{x/d}$ belongs to 
  $\prooftree{\circuit}{\tau}$ so there is at least one $e \in
  \outgoing{u_{x/d}}$ that belongs to $\prooftree{\circuit}{\tau}$. Similarly, if
  $\tau \in \soledge{e}$ for some $e \in \outgoing{u_{x/d}}$ then $u_{x/d}$ is also in
  $\prooftree{\circuit}{\tau}$, and then $\tau(x)=d$. This proves the equality.
  The disjointness of the union directly follows from
  Proposition~\ref{prop:outgoingDisjoint}.
\end{proof}

Next we show that the disjoint union of the relations induced by each ingoing
edge of a \orgate{} is equal to the disjoint union of the relations induced by
each outgoing edge of this gate.

\begin{proposition}
    \label{prop:prooftreeOr}
    Let $u$ be an internal \orgate{} of $\circuit$, \\ 
        \[\disjointsolunion{i}{\ingoing{u}}{\circuit} = \disjointsolunion{o}{\outgoing{u}}{\circuit}.\]
\end{proposition}   
\begin{proof}
  Let $S_u := \{\tau \mid u \in \prooftree{\circuit}{\tau} \}$ and $\tau \in
  S_u$. It is clear from Proposition~\ref{prop:prooftreeUnique} that a proof
  tree contains $u$ if and only if it contains at least an edge of $\ingoing{u}$
  and at least an edge of $\outgoing{u}$. Thus, both unions are equal to $S_u$.

  The disjointness of the first union directly follows from the second item
  of Proposition~\ref{prop:prooftreeUnique} and the disjointness of the
  second union from Proposition~\ref{prop:outgoingDisjoint}.
\end{proof}

Finally we show that the disjoint union of the relations induced by the outgoing edges of 
a \andgate{} is equal to the relation induced by each ingoing edge of this gate.

\begin{proposition}
    \label{prop:prooftreeAnd}
    Let $u$ be an internal \andgate{} of $\circuit$,
        \[\forall i \in \ingoing{u} : \soledge{i} = 
        \disjointsolunion{o}{\outgoing{u}}{\circuit} \]
\end{proposition}

\begin{proof}
  Let $S_u := \{\tau \mid u \in \prooftree{\circuit}{\tau} \}$. By
  Proposition~\ref{prop:prooftreeUnique}, it is clear that if $u$ is in a proof
  tree $\somept$, then every edge of $\ingoing{u}$ is also in $\somept$ and
  exactly one edge of $\outgoing{u}$ is in $\somept$. Thus, both sets in the
  statement are equal to $S_u$. The disjointness of the union follows directly
  from Proposition~\ref{prop:outgoingDisjoint}.
\end{proof}

\subsection{Theorem statement}
\label{sec:th-statement}

Given a tuple-weighting $\om$ of a {\ddcircuit} $\circuit$, it naturally induces
an edge-weighting $W$ defined for every $e \in \edges{C}$ as $W(e) := \sum_{\tau \in
  \soledge{e}} \om(e)$. We call $W$ the \emph{edge-weighting induced by $\om$}.

\begin{theorem}
    \label{th:weightingCorrespondence}
    Given a {\ddcircuit} $\circuit$. 
    \begin{enumerate}[(a)]
    \item\label{it:omtoW} For every tuple-weighting $\om$ of $\circuit$, the
      edge-weighting $\W$ induced by $\om$ is sound.
    \item\label{it:Wtoom} Given a sound edge-weighting $\W$ of $\circuit$, there exists a
      tuple-weighting $\om$ of $\circuit$ such that $W$ is the edge-weighting
      induced by $\om$. 
    \end{enumerate}    
\end{theorem}

We will provide the proof for this theorem in the following sections. 

Observe that if if an edge-weighting $\W$ is induced by a tuple-weighting $\om$ then 
\begin{align*}
\xdtuplesum{x}{d}{\solcircuit{}} 
    = \sum\limits_{e \in \outgoing{x/d}} \sum\limits_{\tau \in \soledge{e}} \val{\om}{\tau}
    = \sum\limits_{e \in \outgoing{x/d}} \val{\W}{e} 
\end{align*}
so $\W$ is also compatible with $\om$.

Thus Theorem~\ref{th:weightingCorrespondence} is indeed a generalization of
Theorem~\ref{th:weightingCompatibility}.

\subsection{Proof of Theorem~\ref{th:weightingCorrespondence}~\ref{it:omtoW}}
\label{sec:tuples-edges}

This section is dedicated to the proof of Theorem~\ref{th:weightingCorrespondence}~\ref{it:omtoW}, 
namely that, given a {\ddcircuit} $\circuit$ and a tuple-weighting $\om$ of $\circuit$, 
the edge-weighting $\W$ of $\circuit$ induced by $\om$ defined as $\val{\W}{e} := \sum_{\tau \in
  \soledge{e}} \om(\tau)$ is sound.

We will see that the soundness of $\W$ follows naturally from the properties of
proof-trees of the previous section. 
We prove that $\W$ is sound by checking the case of \orgate{s} and \andgate{s} separately.

    \begin{itemize}
    \item We first have to show that for every \orgate{} $u$ of $\circuit$, it
      holds that $\sum_{e \in \outgoing{u}} \W(e) = \sum_{e \in \ingoing{u}}
      \W(e)$. This is a consequence of Proposition~\ref{prop:prooftreeOr}. Let
      $u$ be a \orgate{} of $\circuit$.

        \begin{align*}
          \sum_{e \in \outgoing{u}} \W(e)  &  = \sum_{e \in \outgoing{u}} \sum_{\tau \in \soledge{e}} \om(\tau) \proofnote{by definition of $\W$}\\
          & = \sum_{\tau \in R} \om(\tau) \proofnote{where $R = \biguplus_{e \in \outgoing{u}} \soledge{e}$}
        \end{align*}

        The disjointness of the union in $R$ has been proven in Proposition
        \ref{prop:prooftreeOr}, which also states that $R = \biguplus_{e \in
          \ingoing{u}} \soledge{e}$. Thus, the last term in the sum can be
        rewritten as
        \begin{align*}
          \sum_{\tau \in R} \om(\tau)  &  = \sum_{e \in \ingoing{u}} \sum_{\tau \in \soledge{e}} \om(\tau)\\
          & = \sum_{e \in \ingoing{u}} \W(e)  \proofnote{by definition of $\W$}
        \end{align*}


      \item We now show that for every \andgate{} $u$ of $\circuit$ and for
        every edge $i \in \ingoing{u}$ going in $u$, it holds that $\sum_{e \in
          \outgoing{u}} \W(e) = \W(i)$. Let $u$ be a \andgate{} and $i \in
        \ingoing{u}$. The proof is very similar to the previous case but is now
        a consequence of Proposition~\ref{prop:prooftreeAnd}. As before, we have
        $\sum_{e \in \outgoing{u}} \W(e) = \sum_{\tau \in R} \om(\tau)$ where $R
        = \biguplus_{e \in \outgoing{u}} \soledge{e}$. The disjointness of the
        union in $R$ has been proven in Proposition \ref{prop:prooftreeAnd},
        which also implies that $R = \soledge{i}$. Thus, we have:

        \begin{align*}
          \sum_{e \in \outgoing{u}} \W(e) &  = \sum_{\tau \in \soledge{i}} \om(\tau)\\
          & = \W(i)  \proofnote{by definition of $\W$}
        \end{align*}
        

    \end{itemize}

\subsection{Proof of Theorem~\ref{th:weightingCorrespondence}~\ref{it:Wtoom}}
\label{sec:edges-tuples}

This section is dedicated to the proof of Theorem~\ref{th:weightingCorrespondence}~\ref{it:Wtoom}, 
namely that, given a circuit $\W$ and a sound edge-weighting $\W$ of $\circuit$, 
there exists a tuple-weighting $\om$ of $\circuit$ such that, for every edge $e$ of $\circuit$,
$\sum_{\tau \in \soledge{e}} \om(\tau) = \val{\W}{e}$.

In this section, we fix a {\ddcircuit} $\circuit$ and a sound edge-weighting
$\W$ of its edges. 
We assume wlog that the root $r$ of $\circuit$ has a single ingoing edge which we call the \emph{output edge} $\rootedge$. 
We construct $\om$ by induction. For every edge $e =
\edge{u}{v}$ of $\circuit$, we inductively construct a tuple weighting $\om[e]$
of $\solcircuit{u}$. We will then choose $\om = \om[o_r]$ 
and show that this tuple-weighting verifies $\sum_{\tau \in
  \soledge{e}} \om(\tau) = \val{\W}{e}$ for every edge $e$ of $\circuit$.

For $e = \edge{u}{v}$, we define $\om[e] : \solcircuit{u} \rightarrow \R_+$
inductively as follows:
\begin{itemize}
\item \textbf{Case 1}: $u$ is an input labeled with $x/d$ then $\solcircuit{u}$ contains only
  the tuple $x/d$. We define $\om[e](x/d) := W(e)$.
\item \textbf{Case 2}: $u$ is a \orgate{} with children $u_1, u_2$. Let $e_1 =
  \edge{u_1}{u}$ and $e_2 = \edge{u_2}{u}$. In this case, given $\tau \in
  \solcircuit{u}$, we have by definition that $\tau \in \solcircuit{u_1}$ or
  $\tau \in \solcircuit{u_2}$. Assume wlog that $\tau \in \solcircuit{u_1}$. If
  $\sum_{f \in \outgoing{u}} W(f) \neq 0$, we define: $\om[e](\tau) := W(e)
  {\om[e_1](\tau) \over W(e_1)+W(e_2)}$. Observe that since $\W$ is sound,
  $W(e_1)+W(e_2) = \sum_{f \in \outgoing{u}} W(f) \neq 0$. Otherwise
  $\om[e](\tau) := 0$.
\item \textbf{Case 3}: $u$ is a \andgate{} with children $u_1, u_2$. Let $e_1 =
  \edge{u_1}{u}$ and $e_2 = \edge{u_2}{u}$. In this case, given $\tau \in
  \solcircuit{u}$, we have by definition that $\tau = \tau_1 \times \tau_2$ with
  $\tau_1 \in \solcircuit{u_1}$ and $\tau_2 \in \solcircuit{u_2}$. If $W(e_1)
  \neq 0$ and $W(e_2) \neq 0$, we define: $\om[e](\tau) := W(e) {\om[e_1](\tau_1)
    \over W(e_1)}{\om[e_2](\tau_2) \over W(e_2)}$. Otherwise $\om[e](\tau) := 0$.
\end{itemize}




%

We begin our proof by showing that as we construct each $\om[e]$ through a
bottom-up induction there is a relation between $\val{\W}{e}$ and $\om[e]$ which
resembles the property we aim to prove in this section.

\begin{lemma}
    \label{lem:WToOmUpward}
    For every gate $u$ of $\circuit$ and $e = \edge{u}{v} \in \edges{\circuit}$,
    \[\val{\W}{e} = \sum_{\tau \in \solcircuit{u}} \om[e](\tau). \]
\end{lemma}

\begin{proof}

  Let $e=\edge{u}{v}$ be an edge of $\circuit$. Observe that when $\sum_{o \in
    \outgoing{u}} \W(o) = 0$, then $\W(e)=0$ since $e \in \outgoing{u}$ and $\W$
  has positive value. Moreover, by definition of $\om[e]$, for every $\tau \in
  \solcircuit{u}$, $\om[e](\tau) = 0$. In particular, $\sum_{\tau \in
    \solcircuit{u}} \om[e](\tau) = 0 = \W(e)$. In this case then, the lemma
  holds.

  In the rest of the proof, we now assume that $\sum_{o \in \outgoing{u}}
  \W(o) \neq 0$. We show the lemma by induction on the nodes of $\circuit$ from the
  leaves to the root.

  \textbf{Base case}: $u$ is a leaf labelled with $x/d$. Let $e$ be an outgoing
  edge of $u$. Observe that $\solcircuit{u}$ contains a single tuple $x/d$ and
  that, by definition of $\om[e]$, $\val{\W}{e} = \om[e](x/d) = \sum_{\tau \in
    \solcircuit{u}} \om[e](\tau)$.
 
  \textbf{Inductive case}: Now let $u$ be an internal gate of $\circuit$ with
  children $u_1,u_2$ and let $e_1 = \edge{u_1}{u}$ and $e_2 = \edge{u_2}{u}$, as
  depicted in Figure~\ref{fig:inductivestep}.

  \begin{figure}[H]
    \centering    
    \includegraphics[width=2cm]{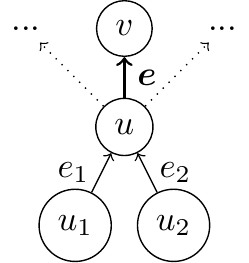}
    \caption{Inductive step notations.}
    \label{fig:inductivestep}
  \end{figure}

  \textbf{Case 1}: Assume that $u$ is a \orgate{}. Let
      $W=W(e_1)+W(e_2)$. Since $u$ is disjoint, given $\tau \in \solcircuit{u}$,
      either $\tau \in \solcircuit{u_1}$ or $\tau \in \solcircuit{u_2}$ but not
      both. It follows:
 
            \begin{align*}
              \sum_{\tau \in \solcircuit{u}} \om[e](\tau) & = \sum_{\tau \in \solcircuit{u_1}} \om[e](\tau) + \sum_{\tau \in \solcircuit{u_2}} \om[e](\tau) \\
              & = \sum_{\tau \in \solcircuit{u_1}}{W(e) \over W} \om[e_1](\tau_1) + \sum_{\tau \in \solcircuit{u_2}} {W(e) \over W}\om[e_2](\tau_2)
            \end{align*}
            by definition of $\om[e]$. Observe that by induction we have $W(e_1)
            = \sum_{\tau \in \solcircuit{u_1}}\om[e_1](\tau_1)$ and $W(e_2) =
            \sum_{\tau \in \solcircuit{u_2}}\om[e_2](\tau_2)$. Thus, by taking
            the constants out and using this identity, it follows:
            \begin{align*}
              \sum_{\tau \in \solcircuit{u}} \om[e](\tau) & =  {W(e) \over W} \sum_{\tau \in \solcircuit{u_1}}\om[e_1](\tau_1)+ {W(e) \over W} \sum_{\tau \in \solcircuit{u_2}}\om[e_2](\tau_2) \\
                                                          & = {W(e) \over W}W(e_1)+{W(e) \over W}W(e_2) \\
                                                          & = W(e) {W(e_1) + W(e_2) \over W} \\
                                                          & = W(e).
            \end{align*}
            
  \textbf{Case 2}: Assume that $u$ is a \andgate{}. 
            Applying the definition of $\om[e]$, we get:
            \begin{align*}
              \sum_{\tau \in \solcircuit{u}} \om[e](\tau) & = \sum_{\tau \in \solcircuit{u}} W(e){\om[e_1](\projt{\tau}{\var{C_{u_1}}}) \over \val{\W}{e_1}}{\om[e_2](\projt{\tau}{\var{C_{u_2}}}) \over  \val{ \W}{e_2}} 
            \end{align*}
            Remember that by definition of {\ddcircuit}s, $\solcircuit{u} =
            \solcircuit{u_1} \times \solcircuit{u_2}$. That is, $\tau \in
            \solcircuit{u}$ if and only if $\tau_1 :=
            \projt{\tau}{\var{C_{u_1}}} \in \solcircuit{u_1}$ and $\tau_2 :=
            \projt{\tau}{\var{C_{u_2}}} \in \solcircuit{u_2}$. Thus, we can
            rewrite the last sum as:
            \begin{align*}
              \sum_{\tau \in \solcircuit{u}} \om[e](\tau) & = \sum_{\tau_1 \in \solcircuit{u_1}}\sum_{\tau_2 \in \solcircuit{u_2}} W(e){\om[e_1](\tau_1) \over W(e_1)}{\om[e_2](\tau_2) \over W(e_2)}
            \end{align*}
            By taking the constant $W(e)$ out of the sum and observing that the
            sum is now separated into two independent terms, we have:
            \begin{align*}
            \sum_{\tau \in \solcircuit{u}} \om[e](\tau) & = W(e)\sum_{\tau_1 \in \solcircuit{u_1}} {\om[e_1](\tau_1) \over W(e_1)}\sum_{\tau_2 \in \solcircuit{u_2}} {\om[e_2](\tau_2) \over W(e_2)}                                                            
            \end{align*}

            By induction, $W(e_i) = \sum_{\tau \in
              \solcircuit{u_i}}\om[e_i](\tau_i)$ for $i=1,2$. Hence both sums of
            the last term are equal to $1$, which means that $$\sum_{\tau \in
              \solcircuit{u}} \om[e](\tau) = W(e).$$
            

\end{proof}

We choose $\om = \om[o_r]$ where $o_r$ is the output edge.
Lemma~\ref{lem:WToOmUpward} is however not enough to prove
Theorem~\ref{th:weightingCorrespondence} as it only gives the equality $W(e) =
\sum_{\tau \in \soledge{e}} \om(\tau)$ for $e = o_r$. Fortunately we can show
that it holds for every edge $e$ of the circuit. We actually prove a stronger
property, that $\om[e]$ is, in some sense, a projection of $\om$.

Given an edge $e = \edge{u}{v}$ of $\circuit$ and $\tau' \in \solcircuit{u}$, we
denote by $\solproj{e}{\tau'}$ the set of tuples $\tau$ of $\soledge{e}$ such
that $\projt{\tau}{\var{C_u}} = \tau'$. We prove the following:

\begin{lemma}
  \label{lem:WToOmDownward}
    For every $e = \edge{u}{v} \in \edges{\circuit}$, for every $\tau' \in \solcircuit{u}$,
    \[\om[e](\tau') = \sum_{\tau \in \solproj{e}{\tau'}} \om(\tau).\]
\end{lemma}

\begin{proof}
  The proof is by top-down induction on $\circuit$. 

  \textbf{Base case}: We prove the result for $e = \rootedge = \edge{u}{v}$. Let
  $\tau' \in \solcircuit{u}$. Because $u$ is the output gate, we have
  $\var{\circuit_u} = \var{\circuit}$ and hence $\solproj{e}{\tau'} =
  \{\tau'\}$. Recall that $\om = \om[{o_r}]$ by definition. In other words, $\om[e](\tau') =
  \om(\tau') = \sum_{\tau \in \solproj{e}{\tau'}} \om(\tau)$.
 
  \textbf{Inductive case}: Now let $e=\edge{u}{v}$ be an internal edge of
  $\circuit$. Let $o_1,\dots,o_n$ be the outgoing edges of $v$, $u'$ be the only
  sibling of $u$ and let $e'=\edge{u'}{v}$. See Figure~\ref{fig:inductive_step2}
  for a schema of these notations. We fix $\tau' \in \solcircuit{u}$ and prove
  the desired equality.

    \begin{figure}[H]
    \centering
    \includegraphics[width=2cm]{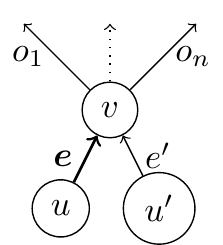}
    \caption{Notations for the inductive step.}
    \label{fig:inductive_step2}
    \end{figure}

    \textbf{Case 1}: $v$ is a \orgate{}. In this case, $\tau' \in
    \solcircuit{v}$. We claim that
    \[\solproj{e}{\tau'} = \biguplus_{o \in \outgoing{v}} \solproj{o}{\tau'}.\]
    For left-to-right inclusion, let $\tau \in \solproj{e}{\tau'}$. By
    definition, its proof tree $\prooftree{\circuit}{\tau}$ contains $e$. Since
    $\prooftree{\circuit}{\tau}$ is connected by
    Proposition~\ref{prop:prooftreeUnique}, $\prooftree{\circuit}{\tau}$ has to
    contain one edge of $\outgoing{v}$. The disjointness of
    the right-side union is a direct consequence of
    Proposition~\ref{prop:outgoingDisjoint}. For the right-to-left inclusion,
    fix $o \in \outgoing{v}$ and let $\tau \in \solproj{o}{\tau'}$. By
    definition, its proof tree $\prooftree{\circuit}{\tau}$ contains $o$, thus
    it also contains the vertex $v$. Now recall that $\projt{\tau}{\var{C_u}} =
    \tau' \in \solcircuit{u}$. Thus, by definition of proof trees, $u$ is also
    in $\prooftree{\circuit}{\tau}$. In other words, $\tau \in
    \solproj{e}{\tau'}$. Using this equality, we have:
    \begin{align*}
      \sum_{\tau \in \solproj{e}{\tau'}} \om(\tau) & = \sum_{o \in \outgoing{v}} \sum_{\tau \in \solproj{o}{\tau'}} \om(\tau)\\
      &= \sum_{o \in \outgoing{v}} \om[o](\tau'). \label{eq:test}
    \end{align*}
    Since, by induction, we have that for every $o \in \outgoing{v}$,
    $\om[o](\tau')=\sum_{\tau \in \solproj{o}{\tau'}} \om(\tau)$.

    Assume first that $\sum_{o \in \outgoing{v}} W(o) = 0$. In this case, by
    definition, for every $o$, $\om[o](\tau') = 0$. Since $W$ is sound however,
    we also have $W(e) = 0$, which implies by Lemma~\ref{lem:WToOmUpward}, that
    $\om[e](\tau')=0$ as well. In this case, $\sum_{\tau \in \solproj{e}{\tau'}}
    \om(\tau) = 0 = \om[e](\tau')$ which is the induction hypothesis.

    Now assume that $\sum_{o \in \outgoing{v}} W(o) \neq 0$. We can thus apply the
    definition of $\om[o](\tau') = {W(o) \over W(e)+W(e')} \om[e](\tau')$ in the
    last sum. It gives
    \begin{align*}
      \sum_{o \in \outgoing{v}} \om[o](\tau') &= \sum_{o \in \outgoing{v}}{W(o) \over W(e)+W(e')} \om[e](\tau') \\ 
                                              & = \om[e](\tau') {1\over W(e)+W(e')}\sum_{o \in \outgoing{v}}{W(o)} \\
                                              & = \om[e](\tau')
    \end{align*}
    where the last equality follows from the fact that $\W$ is sound and thus
    the ratio is $1$.


    \textbf{Case 2}: $v$ is a \andgate{}. Similarly as before, we have:

       \[\solproj{e}{\tau'} = \biguplus_{\tau'' \in \solcircuit{u'}}\biguplus_{o \in \outgoing{v}} \solproj{o}{\tau'
           \times \tau''}.\]
    
    For left-to-right inclusion, let $\tau \in \solproj{e}{\tau'}$. By
    definition, its proof tree $\prooftree{\circuit}{\tau}$ contains $e$. Since
    $\prooftree{\circuit}{\tau}$ is connected by
    Proposition~\ref{prop:prooftreeUnique}, $\prooftree{\circuit}{\tau}$ has to
    contain one edge $o$ of $\outgoing{v}$. Thus, $\tau \in \solproj{o}{\tau'
      \times \projt{\tau}{\var{C_{u'}}}}$.

    The disjointness of the right-side union is a direct consequence of
    Proposition~\ref{prop:outgoingDisjoint}. For the right-to-left inclusion,
    fix $o \in \outgoing{v}$ and $\tau'' \in \solcircuit{u'}$. Let $\tau \in
    \solproj{o}{\tau' \times \tau''}$. By definition, its proof tree
    $\prooftree{\circuit}{\tau}$ contains $o$, thus it also contains the vertex
    $v$ and by definition of proof trees, $u$ is also in
    $\prooftree{\circuit}{\tau}$. And since $\projt{\tau}{\var{C_u}} = \tau'$,
    $\tau \in \solproj{e}{\tau'}$. Using this equality, we have:
    \begin{align*}
      \sum_{\tau \in \solproj{e}{\tau'}} \om(\tau) & = \sum_{o \in \outgoing{v}} \sum_{\tau'' \in \solcircuit{u'}} \sum_{\tau \in \solproj{o}{\tau' \times \tau''}} \om(\tau)\\
      &= \sum_{o \in \outgoing{v}} \sum_{\tau'' \in \solcircuit{u'}} \om[o](\tau' \times \tau''). 
    \end{align*}
    Since, by induction, we have that for every $o \in \outgoing{v}$,
    $\om[o](\tau' \times \tau'')=\sum_{\tau \in \solproj{o}{\tau' \times \tau''}} \om(\tau)$.

    Assume first that $\sum_{o \in \outgoing{v}} W(o) = 0$. In this case, by
    definition, for every $o$ and $\tau''$, $\om[o](\tau' \times \tau'') = 0$.
    Since $W$ is sound however, we also have $W(e) = 0$, which implies by
    Lemma~\ref{lem:WToOmUpward}, that $\om[e](\tau')=0$ as well. In this case,
    $\sum_{\tau \in \solproj{e}{\tau'}} \om(\tau) = 0 = \om[e](\tau')$ wich is
    the induction hypothesis.

    Now assume that $\sum_{o \in \outgoing{v}} W(o) \neq 0$. We can thus apply
    the definition of $\om[o](\tau' \times \tau'') = W(o) {\om[e](\tau') \over
      W(e)}{\om[e'](\tau'') \over W(e')}$ in the last sum. It gives
    \begin{align*}
      & \sum_{o \in \outgoing{v}} \sum_{\tau'' \in \solcircuit{u'}} \om[o](\tau' \times \tau'') \\
      & = \sum_{o \in \outgoing{v}} \sum_{\tau'' \in \solcircuit{u'}} W(o) {\om[e](\tau') \over
      W(e)}{\om[e'](\tau'') \over W(e')} \\
      & = {\om[e](\tau')\over W(e)} \big(\sum_{o \in \outgoing{v}} W(o)\big) {\sum_{\tau'' \in \solcircuit{u'}} \om[e'](\tau'') \over W(e')}
    \end{align*}

    Since $W$ is sound, $\sum_{o \in \outgoing{v}} W(o)=W(e)$. Moreover, by
    Lemma~\ref{lem:WToOmUpward}, $\sum_{\tau'' \in \solcircuit{u'}}
    \om[e'](\tau'')=W(e')$. Thus, the last sum equals to $\om[e](\tau')$ which
    concludes the proof.
    
\end{proof}

Theorem~\ref{th:weightingCorrespondence}~\ref{it:Wtoom} is now an easy
consequence of Lemma~\ref{lem:WToOmUpward} and Lemma~\ref{lem:WToOmDownward}:
\begin{align*}
  W(e) & = \sum_{\tau' \in \solcircuit{u}} \om[e](\tau') \proofnote{by Lemma~\ref{lem:WToOmUpward}}\\
       & = \sum_{\tau' \in \solcircuit{u}} \sum_{\tau \in \solproj{e}{\tau'}} \om(\tau)\proofnote{by Lemma~\ref{lem:WToOmDownward}}\\
       & = \sum_{\tau \in \soledge{e}} \om(\tau) \proofnote{since $\soledge{e} = \biguplus_{\tau' \in
         \solcircuit{u}} \solproj{e}{\tau'}$}.
\end{align*}

Indeed, $\soledge{e} = \biguplus_{\tau' \in \solcircuit{u}} \solproj{e}{\tau'}$.
For the left-to-right inclusion, if $\tau \in \soledge{e}$ then by definition,
$\tau' = \projt{\tau}{\var{C_u}} \in \solcircuit{u}$ and thus $\tau \in
\solproj{e}{\tau'}$. The other inclusion follows by definition since
$\solproj{e}{\tau'} \subseteq \soledge{e}$ for every $\tau' \in \solcircuit{u}$.

\section{Consequences and extension of the result}
\label{sec:proj}

\subsection{Effective reconstruction of solutions}

As it is illustrated in the motivating example of the introduction, one may not
only be interested in the optimal value of a {\sumlp} but also on finding an
optimal solution. As the size of a factorized relation may be too big, it may
not be tractable to output the entire optimal solution. However, the
construction used in the proof of Theorem~\ref{th:weightingCorrespondence} to go
from an edge-weighting to a tuple-weighting is effective in the sense that given
an edge-weighting $W$ of a {\ddcircuit} $C$ and $\tau \in \solcircuit{}$, one
can compute $\om(\tau)$ in time polynomial in $|C|$ (where $\om$ is the
tuple-weighting inducing $W$ given by Theorem~\ref{th:weightingCorrespondence}).
It is indeed enough to compute the partial tuple weightings $\omega_e$ presented
in Section~\ref{sec:edges-tuples} in a bottom-up induction on $\circuit$.

\subsection{Tractability for conjunctive queries}

By connecting Theorem~\ref{th:rewritecaslp} and Theorem~\ref{thm:cq-to-fr}, we
directly get the tractability of solving bounded fractional hypertree width
quantifier free conjunctive queries:

\begin{corollary}
  \label{cor:rewritecaslp_cq} Given a {\sumlp} $L$, a quantifier-free
  conjunctive query $Q$, a hypertree decomposition of $Q$ of fractional
  hypertreewidth $k$ and a database $\db$, one can compute the optimal value of
  $L(Q(\db))$ in polynomial time.
\end{corollary}

The exact complexity of Corollary~\ref{cor:rewritecaslp_cq} depends on the
runtime of the linear program solver that is used that also highly depends on
the structure of the linear program. We do not know whether we could directly
solve the linear program on the circuit without calling the solver and have
better complexity bounds.

The complexity also depends on the fact that we are given a good hypertree
decomposition in the input. Computing the best decomposition for fractional
hypertree width is known to be an $\NP$-hard problem~\cite{fischl2018general}
but it can be approximated in polynomial time to a cubic factor~\cite{marx10}
which is enough if one is interested only in theoretical tractability. In
practice however, our algorithm would perform the best when the circuit is
small, be it computed from hypertree decompositions or other techniques.

Our result does not directly apply to quantified conjunctive queries. Indeed,
existentially or universally projecting variables in a {\ddcircuit} may lead to
an exponential blow-up of its size or the disjointness of unions may be lost.
However, for existential quantification, we do not need to actually project the
variables in the circuit to compute the optimal value. Given a relation $R$ on
attribute $\varx$ and $Z \subseteq \varx$, we denote by $\exists Z. R =
\{\tau|_{\varx \setminus Z} \mid \tau \in R\}$. We have the following:

\begin{theorem}
  \label{th:caslpproj}
  Let $L$ be a {\sumlp} on attributes $\varx \setminus Z$ and domain $\dom$. Let
  $R$ be a relation on attributes $\varx$ and domain $\dom$. $L(R)$ and
  $L(\exists Z.R)$ have the same optimal value.
\end{theorem}
\begin{proof}
  We show how one can transform a solution of $L(R)$ into a solution of
  $L(\exists Z.R)$ with the same value and vice-versa.

  We introduce the notation for $\tau' \in \exists Z.R$:
  \[\Ext(\tau') = \{\tau'' : Z \rightarrow \dom \mid \tau'\times\tau'' \in R\}.\]
  
  Let $\om : R \rightarrow \R_+$ be a solution of $L(R)$. Given $\tau' \in
  \exists Z.R$, we define $\om'(\tau') := \sum_{\tau'' \in \Ext(\tau')}
  \om(\tau' \cdot \tau'')$. It is easy to see that $\om'$ is a solution of $L(\exists Z.R)$
  since $\om'(S_{x,d}) = \om(S_{x,d})$ for every $x \in \varx \setminus Z$ and
  $d \in \dom$.

  Similarly, if $\om' : \exists Z. R \rightarrow \R_+$ is a solution of
  $L(\exists Z. R)$, we define for $\tau \in R$, $\om(\tau) := {\om'(\tau')
    \over \#\Ext(\tau')}$ where $\tau' = \tau|_{\varx \setminus Z}$. Again, it
  is easy to see that $\om(S_{x,d}) = \om'(S_{x,d})$ which concludes the proof.
\end{proof}

Theorem~\ref{th:caslpproj} implies that one can still solve {\sumlp}s on
relations given as existential projection of a {\ddcircuit}. In other words, if
a conjunctive query $Q$ is of the form $\exists Z. Q'$ and $Q'$ has bounded
fractional hypertree width, then we can still solve {\sumlp}s on $Q(\db)$ for a
given database $\db$ in polynomial time by transforming $Q'$ into a {\ddcircuit}
and solving the {\sumlp} on it directly.

\section{Conclusion}
\label{sec:conclusion}


In this paper, we have initiated the study of a new kind of natural aggregation
tasks: solving optimization problems whose variables are the answers of a
database query. We isolated an interesting class of linear programs that is
intractable for conjunctive queries in general but for which we provide an
algorithm that runs in polynomial time when the fractional hypertreewidth of the
conjunctive query is bounded. Our technique relies on factorized representation
of the answer set of the conjunctive query.

There are many possible future directions based on the techniques we present in
this paper. A first step would be to generalize our methods to larger classes of
linear programs, allowing more complex meta-variables than the {\sxdname}s that
we use in this paper. Another limitation of our work that we would like to
improve is that it can currently only handle linear programs having positive
values for technical reasons in the way our proof is constructed. Building on
this idea, a very exciting yet ambitious program is to generalize and integrate
our methods in a larger framework resembling GNU Math Prog or AMPL augmented
with database queries on which we could act to generate smaller linear programs
than the naive grounding.

Finally, we would like to investigate applications of our result and methods to
existing techniques in data mining and databases. A first interesting step would
be to use our algorithm to compute the $s$-measure
of~\cite{wang_efficiently_2013} in practice for some real life queries. We would
like to try to run the algorithm on benchmarks from data mining and see how they
compare with the bruteforce approach of generating the complete ground linear
program. We already tried a basic implementation of our algorithm on syntactical
data to see if the linear program solver scales but did not yet consider any
real benchmark.

Moreover, it would be interesting to investigate more deeply connections with
other works using Binary or Mixed-Integer Programming to solve database
problems. In \cite{kolaitis_efficient_2013} the authors use binary programming
to answer queries on inconsistent databases. They do so by generating binary
programs with two types of constraints. The first type of constraints encodes
the possible repairs of the database by limiting the total weight of the sum of
the tuples sharing a value for a primary key. These are constraints over
{\sxdname}s in our framework and can thus be handled seamlessly. The second type
of constraints deals with witnesses of the potential answers. We might be able
to handle theses constraints by adding some information about the witnesses to
the circuit.

In \cite{tiresias} the authors use Mixed-Integer Programming to answer how-to
queries expressed in the "Tiresias Query Language" (TiQL) which is based on
datalog. The constraints of their Mixed-Integer Programs appear to be quite
different at first glance from what our framework is able to handle. However the
constraints are generated using the provenance of the tuples, a notion which is
closely related to the circuits we exploit to solve \sumlp s more efficiently so
it might be possible to exploit the structure of the circuits in a different way
to answer TiQL queries more efficiently.

\paragraph{Acknowledgment}
This work was partially supported by the French Agence Nationale de la
Recherche, AGGREG project reference ANR-14-CE25-0017-01, Headwork project
reference ANR-16-CE23-0015 and by a grant of the Conseil Régional
Hauts-de-France. The project DATA, Ministère de l'Enseignement Supérieur et de
la Recherche, Région Nord-Pas de Calais and European Regional Development Fund
(FEDER) are acknowledged for supporting and funding this work. We also thank Rui
Castro, Sylvain Salvati, Sophie Tison and Yuyi Wang for fruitful discussions and
anonymous reviewers of a previous version of this paper for their helpful
comments.

\bibliography{biblio}

\clearpage

\appendix

\end{document}